\documentclass[11pt]{article}
\usepackage{fullpage,epsf,amssymb,amsthm,amsfonts,amsmath,latexsym, graphicx,array,extarrows,mathtools,mathstyle}
\usepackage[font=footnotesize,labelfont=bf]{caption}
\usepackage[normalem]{ulem}
\usepackage{float}

\newcommand{\overbar}[1]{\mkern 1.5mu\overline{\mkern-1.5mu#1\mkern-1.5mu}\mkern 1.5mu}

\let\oldsqrt\sqrt
\def\sqrt{\mathpalette\DHLhksqrt}
\def\DHLhksqrt#1#2{%
\setbox0=\hbox{$#1\oldsqrt{#2\,}$}\dimen0=\ht0
\advance\dimen0-0.2\ht0
\setbox2=\hbox{\vrule height\ht0 depth -\dimen0}
{\box0\lower0.4pt\box2}}

\title{\bf{Non-manifest symmetries in quantum field theory}} 
\author{\large{Peter Lowdon} \\
\ \\
\textit{\small{Physik-Institut, Universit\"at Z\"urich, Winterthurerstrasse 190, 8057 Z\"urich, Switzerland}} \\
\textit{\small{E-mail: lowdon@physik.uzh.ch}}}
\date{}
\begin{document}
\begin{flushright} ZU--TH 32/15   \end{flushright}
\vspace{15mm} 
{\let\newpage\relax\maketitle}
\setcounter{page}{1}
\pagestyle{plain}

\maketitle
\setcounter{page}{1}
\pagestyle{plain}
 
\abstract \noindent 
Non-manifest symmetries are an important feature of quantum field theories (QFTs), and yet their characteristics are not fully understood. In particular, the construction of the charge operators associated with these symmetries is ambiguous. In this paper we adopt a rigorous axiomatic approach in order to address this issue. It turns out that charge operators of non-manifest symmetries are not unique, and that although this does not affect their property as generators of the corresponding symmetry transformations, additional physical input is required in order to determine how they act on states. Applying these results to the examples of spacetime translation and Lorentz symmetry, it follows that the assumption that the vacuum is the unique translationally invariant state is sufficient to uniquely define the charges associated with these symmetries. In the case of supersymmetry though there exists no such physical requirement, and this therefore implies that the supersymmetric charge, and hence the supersymmetric space of states, is not uniquely defined.

\newpage

\newtheorem{mydef}{Definition}
\newtheorem{theorem}{Theorem}
\newtheorem{lemma}{Lemma}
\newtheorem{corr}{Corollary}
\newcommand{\lpipe}{\rule[-0.4ex]{0.90pt}{2.3ex}}
\setlength\fboxrule{0.5mm}
\setlength\fboxsep{4mm}

\section{Introduction}
\label{section1}

Non-manifest symmetry is an important feature of quantum field theory (QFT), and yet the quantisation of these types of symmetries is still not fully understood. In particular, the question of which features of classical non-manifest symmetries survive quantisation remains unresolved. Given a classical field theory with fields $\{\varphi_{a}\}$ and Lagrangian density $\mathcal{L}$, a non-manifest symmetry is a symmetry for which $\delta \mathcal{L} = \partial_{\mu}\mathcal{B}^{\mu}$. The corresponding Noether current\footnote{The Noether current may also have additional indices other than $\mu$, but here we will use $j^{\mu}$ for simplicity.} has the form:
\begin{align}
j^{\mu} = \frac{\partial \mathcal{L}}{\partial(\partial_{\mu}\varphi_{a})}\delta\varphi_{a} - \mathcal{B}^{\mu}
\label{nm_current}
\end{align} 
and its conservation $\partial_{\mu}j^{\mu}=0$ follows from the Euler-Lagrange equations. Manifest symmetries are defined by the fact that $\mathcal{B}^{\mu}\equiv 0$, and so the current is canonically determined~\cite{QFT_Maths99}. But for non-manifest symmetries $\mathcal{B}^{\mu}$ is not defined to vanish, which means that there is the freedom to redefine $\mathcal{B}^{\mu}$ by performing the following \textit{improvement} transformation\footnote{$\partial_{\nu}\widetilde{\mathcal{B}}^{[\mu\nu]}$ is referred to as an \textit{improvement term}~\cite{Shifman12}, and $[\mu\nu]$ indicates that the indices $\mu$ and $\nu$ are anti-symmetric.}:
\begin{align}
\mathcal{B}^{\mu}  \longrightarrow \mathcal{B}^{\mu} + \partial_{\nu}\widetilde{\mathcal{B}}^{[\mu\nu]} 
\end{align} 
without changing $\delta \mathcal{L}$ or affecting the conservation of the current. However, because the current $j^{\mu}$ \textit{is} modified, this freedom implies an ambiguity in the definition of $j^{\mu}$~\cite{QFT_Maths99}. Since the charges associated with these currents are defined by $Q= \int d^{3}x \, j^{0}(x)$, improvement terms in the current become spatial boundary terms $\int d^{3}x \, \partial_{i}\widetilde{\mathcal{B}}^{[0i]}$ in the charges. So if two currents $j^{\mu}$ and $\tilde{j}^{\mu}$ differ by an improvement term, it follows from Stokes' theorem that the corresponding charges $Q$ and $\widetilde{Q}$ may be equal if the fields satisfy certain boundary conditions, such as the requirement to vanish at spatial infinity. \\

\noindent
A significant difference between classical and quantum field theories is that quantised fields are operator-valued distributions, and not functions. An important consequence of this is that quantised fields are not point-wise defined~\cite{Streater_Wightman64}. The physical motivation for this feature is that because operators inherently imply a measurement, if a field $\varphi(x)$ \textit{were} a well-defined operator then this would represent the performance of a measurement at a single spacetime point $x$. However, quantum mechanically this is not possible since this would require an infinite amount of energy~\cite{Haag96}. Instead, one can perform a measurement over a spacetime region $\mathcal{U}$ and model this with the operator $\varphi(f) := \int d^{4}x \ \varphi(x)f(x)$, which consists of a distribution $\varphi$ smeared with some test function $f$ with support in $\mathcal{U}$. Quantised fields being operator-valued distributions is one of a series of axioms which are employed in axiomatic approaches to QFT. Although different axiomatic schemes have been proposed, these schemes generally consist of a common core set of axioms including, for example, local (anti-)commutativity and the uniqueness of the vacuum\footnote{See e.g.~\cite{Streater_Wightman64,Haag96,Nakanishi_Ojima90} for a more in-depth discussion.}. The requirement that these axioms are compatible with the definition of a quantised field is also another strong motivation for why these fields are distributions as opposed to functions\footnote{Assuming fields $\varphi(x)$ are operator-valued functions, and combining this with certain standard QFT axioms, implies $\exists c \in \mathbb{C}$ such that: $\varphi(x)|0\rangle = c|0\rangle \ \forall x$, and hence the fields cannot be non-trivial~\cite{Butterfield_Earman07}.}. Because of the distributional behaviour of quantised fields, the arguments used to determine the characteristics of classical non-manifest symmetries are generally no longer valid for QFTs. In particular, since the spatial boundary terms $\int d^{3}x \, \partial_{i}\widetilde{\mathcal{B}}^{[0i]}$ are operators in the quantised theory, and imposing boundary conditions on quantised fields is ill-defined~\cite{Lowdon14}, the classical reasoning used to justify when these terms vanish does not apply. Nevertheless, these arguments are still often cited to explain the vanishing of these terms in QFT~\cite{Weinberg95, Mandl_Shaw10, Greiner_Reinhardt96}.  \\

\noindent
The remainder of this paper is structured as follows: in section~\ref{section2} we analyse the quantisation and general characteristics of non-manifest symmetries in QFTs; in section~\ref{section3} we apply these findings to specific examples of these symmetries; and finally in section~\ref{section4} we summarise our results.

\section{Quantum non-manifest symmetries}
\label{section2}

\subsection{Quantisation}
\label{nm_quant}

As outlined in section~\ref{section1}, spatial boundary operators are an important feature of quantum non-manifest symmetries. In classical field theories, these spatial boundary terms have the form $B = \int d^{3}x \ \partial_{i}B^{i}(x)$, where $B^{i}(x)$ is some field. However, since \textit{quantised} fields are distributions, the quantum analogue of $B$ must involve a smearing with test functions in order to yield a well-defined operator. Moreover, it is necessary that this smearing is performed both in space \textit{and} time, since $\partial_{i}B^{i}(x)$ is in general not defined at sharp times~\cite{Strocchi13}. The spatial boundary operator $B$ can thus be written: $B = \int d^{4}x \ f(x) \partial_{i}B^{i}(x) =  \partial_{i}B^{i}(f)$, where $f$ is some test function on $\mathbb{R}^{1,3}$. As discussed in section~\ref{section1}, the operator $B = \partial_{i}B^{i}(f)$ represents the performance of a measurement on the spacetime region $\text{supp}(f) \subset \mathbb{R}^{1,3}$. Therefore, in order to ensure that $B$ agrees with the classically-motivated form for $B$, and is defined on the whole of space, one can choose without loss of generality~\cite{Lowdon14} that $f(x):=\alpha(x_{0})f_{R}({\bf{x}})$, where $\alpha \in \mathcal{D}(\mathbb{R})$ ($\text{supp}\left(\alpha\right) \subset [-\delta,\delta]$, $\delta >0$) and $f_{R} \in \mathcal{D}(\mathbb{R}^{3})$ have the following properties:
\begin{align}
\int dx_{0} \ \alpha(x_{0}) = 1, \hspace{10mm}
   f_{R}({\bf{x}}) = \left\{
     \begin{array}{ll}
       1, & |{\bf{x}}| < R \\
       0, & |{\bf{x}}| > R(1+\varepsilon)
     \end{array}
     \right.
\label{test_f}
\end{align}
with $\varepsilon >0$. Hence the spatial boundary operator $B$ has the explicit form:
\begin{align}
B = \lim_{R \rightarrow \infty} \int d^{4}x \ \alpha(x_{0})f_{R}({\bf{x}}) \, \partial_{i}B^{i}(x) \,.
\label{B_op}
\end{align}
As well as spatial boundary operators, this same class of test functions can also be used to rigorously define the quantum variation $\delta F$ of a (smeared) field operator $F$ under a symmetry transformation. Given that the symmetry gives rise to the conserved Noether current $j^{\mu}$, $\delta F$ is defined as follows~\cite{Kastler_Robinson_Swieca66, Ferrari_Picasso_Strocchi77, Kugo_Ojima79}:
\begin{align}
\delta F = i\left[Q, F \right]_{\pm} := \lim_{R\rightarrow \infty} i\left[Q_{R}, F \right]_{\pm} = \lim_{R\rightarrow \infty} i\left[j^{0}(\alpha f_{R}), F \right]_{\pm}
\label{var_def}
\end{align}   
where $Q_{R}$ is a localised expression for the charge generator $Q$ of the symmetry, and $\left[\cdot,\cdot \right]_{\pm}$ is either an anti-commutator or commutator depending on the spin of $Q$ and $F$. \\

\noindent
It turns out that the definitions for both $B$ and $\delta F$ are not completely sufficient to ensure that these operators are always well-defined. One also requires that the algebra of fields $\mathcal{F}$ in the theory is \textit{local}, which means that for any fields $\phi,\psi \in \mathcal{F}$, one has that $[\phi(f),\psi(g)]_{\pm}=0$ when the supports of the test functions $f$ and $g$ are space-like separated\footnote{This property is called local (anti-)commutativity.}. However, the locality of $\mathcal{F}$ is not guaranteed for all classes of QFTs. In particular, for gauge theories it transpires that the gauge symmetry of the theory implies a strengthened form of the Noether current conservation condition, and this leads to the possibility of non-local fields~\cite{Strocchi13}. A prominent example of this is quantum electrodynamics (QED) in the Coulomb gauge, where all the \textit{charged} fields are non-local~\cite{Strocchi_Wightman74}. Nevertheless, it turns out the locality of the field algebra can be preserved, and this can be achieved by adopting a so-called \textit{local quantisation}~\cite{Strocchi13}. In local quantisations, additional degrees of freedom are introduced into the theory, resulting in an extension of the space of states $\mathcal{V}$. However, an important consequence of this extension is that the inner product in $\mathcal{V}$ is no longer positive definite. So the locality of $\mathcal{F}$ is preserved at the expense of violating the positivity of the inner product in $\mathcal{V}$. Since negative norm states are unphysical, one must therefore introduce a condition in order to determine the physical states $\mathcal{V}_{\text{phys}} \subset \mathcal{V}$. For Yang-Mills theories, \textit{BRST quantisation} is an important example of a local quantisation. In this case ghost and gauge-fixing degrees of freedom are added to the theory in order to break the gauge invariance, and preserve the locality of $\mathcal{F}$. Although the gauge-fixed theory is no longer gauge invariant, it remains invariant under a residual \textit{BRST} symmetry, with a corresponding conserved charge $Q_{B}$. Physical states are specified by the \textit{subsidiary} condition: $Q_{B}\mathcal{V}_{\text{phys}}=0$, and the corresponding Hilbert space is defined by $\mathcal{H}:=\overbar{\mathcal{V}_{\text{phys}}\slash \mathcal{V}_{0}}$, where $\mathcal{V}_{0} \subset \mathcal{V}_{\text{phys}}$ contains the zero norm states\footnote{The bar implies that $\mathcal{H}$ also includes certain limit states~\cite{Nakanishi_Ojima90}.}. \\

\noindent
Determining the conditions under which spatial boundary operators vanish is central to understanding the differences between classical and quantum non-manifest symmetries. This issue was first investigated for locally quantised QFTs in~\cite{Lowdon14}, where spatial boundary operators $B$ were rigorously defined as in Eq.~(\ref{B_op}). It was established that if $B$ annihilates the vacuum state, then this is both necessary and sufficient to ensure that $B=0$, and hence one has the following theorem: \\ 
\begin{theorem}
 $\int d^{3}x \ \partial_{i}B^{i}$ \hspace{1mm} vanishes in $\mathcal{V}$ \hspace{3mm} $\Longleftrightarrow$ \hspace{3mm} $\int d^{3}x \ \partial_{i}B^{i} | 0 \rangle =0$
\label{theorem1}
\end{theorem} 
\ \\
\noindent
where $\int d^{3}x \ \partial_{i}B^{i}$ implicitly involves the smearing in Eq.~(\ref{B_op}). It should be noted that Theorem~\ref{theorem1} differs slightly to the theorem derived in~\cite{Lowdon14} in that it involves the state space $\mathcal{V}$, as opposed to $\mathcal{H}$. This subtle modification is necessary in order to ensure that the theorem is applicable to arbitrary locally quantised theories\footnote{The requirement for this modification arises because the Reeh-Schlieder theorem, which is central to the proof of Theorem~\ref{theorem1}, holds in $\mathcal{V}$ but may no longer hold in $\mathcal{H}$, as discussed in~\cite{Nakanishi_Ojima90}.}. As already discussed in section~\ref{section1}, given a classical non-manifest symmetry with canonical Noether current $j^{\mu}$, one can modify $j^{\mu}$ by adding an improvement term $\partial_{\nu}\widetilde{\mathcal{B}}^{[\mu\nu]}$ without affecting its overall conservation. The corresponding charges $Q$ and $\widetilde{Q}$ associated with $j^{\mu}$ and $\tilde{j}^{\mu} = j^{\mu} + \partial_{\nu}\widetilde{\mathcal{B}}^{[\mu\nu]}$ will therefore differ by a spatial boundary term, the vanishing of which will depend on the boundary conditions of the (classical) fields. In light of Theorem~\ref{theorem1}, it follows that spatial boundary \textit{operators} are similarly not guaranteed to vanish, and this immediately implies the corollary: \\
\begin{corr}
The charge generator of a quantum non-manifest symmetry is a priori non-unique
\label{nm_corr1}
\end{corr}  
\noindent
Corollary~\ref{nm_corr1} runs contrary to the expectation of much of the established literature~\cite{Shifman12, Weinberg95, Mandl_Shaw10, Greiner_Reinhardt96}, and has to our knowledge not been discussed before. This is in part because it is often incorrectly concluded that spatial boundary \textit{operators} can be assumed to vanish by imposing suitable \textit{classical} boundary conditions. The non-uniqueness of a charge operator $Q$ immediately implies that its action on states is ambiguous, since one can in principle always add an improvement term to the current such that the transformed charge $\widetilde{Q}$ is different to $Q$. This therefore provides a non-perturbative obstacle to the consistency of QFTs that are invariant under a non-manifest symmetry. \\  

\noindent
Before discussing the consequences of Corollary~\ref{nm_corr1} for specific examples of quantum non-manifest symmetries, we will first explore the effect that the ambiguity in the charge has on the generation of the transformations associated with these symmetries. By using the definition of the quantum variation $\delta F$ (Eq.~(\ref{var_def})), one has the following theorem:   
\begin{theorem}
If \ $\widetilde{Q} = Q + \int d^{3}x \ \partial_{i}B^{i}$, where $Q$ is a charge operator, then
\begin{align*}
\widetilde{\delta} F:= i[\widetilde{Q}, F ]_{\pm} = i[Q, F ]_{\pm} = \delta F 
\end{align*} 
for all operators $F$ constructed from (local) fields smeared with some test function  
\label{theorem_op_gen}
\end{theorem}
\begin{proof}
Let $|\Psi\rangle$ and $|\Phi\rangle$ be any arbitrary states, then
\begin{align*}
\langle \Psi | \left( [\widetilde{Q}, F ]_{\pm} - [Q , F ]_{\pm}  \right)|\Phi\rangle &= \langle \Psi | \int d^{3}x  \left[\partial_{i}B^{i}, F \right]_{\pm} |\Phi\rangle \\
&= \langle \Psi |\lim_{R \rightarrow \infty} \int d^{4}x \ \alpha(x_{0})f_{R}({\bf{x}})\left[\partial_{i}B^{i}(x),F \right]_{\pm} |\Phi\rangle\\
&= - \langle \Psi | \lim_{R \rightarrow \infty} \int d^{4}x \ \alpha(x_{0})\left(\partial_{i}f_{R}({\bf{x}})\right)\left[B^{i}(x),F \right]_{\pm} |\Phi\rangle \\
&= - \langle \Psi | \lim_{R \rightarrow \infty} \ \left[B^{i}(\alpha\partial_{i}f_{R}),F \right]_{\pm} |\Phi\rangle =0
\end{align*}
where the vanishing in the last equality follows because  the support of $\alpha \partial_{i}f_{R}$ and the test function in the smearing of $F$ will become space-like separated in the limit $R \rightarrow \infty$, and both $B^{i}$ and $F$ satisfy local (anti-)commutativity.
\end{proof}
\noindent
Theorem~\ref{theorem_op_gen} implies that despite the ambiguity in the definition of the generators of quantum non-manifest symmetries, different charge operators are guaranteed to generate \textit{the same} symmetry transformation. This is in contrast to the classical case, where invariance of the symmetry transformation\footnote{For classical field theories, Poisson brackets are instead used to define the symmetry transformation $\delta F$.} requires that the boundary conditions of the fields must be such that \textit{all} spatial boundary terms are exactly vanishing.  \\

\noindent
Although Theorem~\ref{theorem_op_gen} guarantees that different expressions for the charges of quantum non-manifest symmetries will generate the same transformation,   
Corollary~\ref{nm_corr1}, as outlined previously, implies that the action of charges on states themselves is potentially ambiguous. However, as a consequence of Theorem~\ref{theorem1}, one has the corollary: \\
\begin{corr}
If \ $\widetilde{Q} = Q  + \int d^{3}x \ \partial_{i}B^{i}$ for charge operators $\widetilde{Q}$ and $Q$, then
\begin{align*}
\widetilde{Q}|0\rangle = Q|0\rangle  \hspace{3mm}  \Longleftrightarrow \hspace{3mm} \widetilde{Q}=Q
\end{align*}
\label{Q_vacuum}
\end{corr}
\noindent
Corollary~\ref{Q_vacuum} follows immediately from the fact that $\int d^{3}x \ \partial_{i}B^{i}$ vanishes if and only if it annihilates the vacuum state. Therefore, if any two charges that are related by an improvement transformation act in the same manner on the vacuum state, this is both necessary and sufficient to imply that these charges are equal. In other words, knowledge of how the charge operators of non-manifest symmetries act on the vacuum is sufficient to uniquely define them. This means that in contrast to manifest symmetries, where the conserved current and hence the charge $Q$ are canonically determined (i.e. $\mathcal{B}^{\mu} \equiv 0$ in Eq.~(\ref{nm_current})), quantum non-manifest symmetries require additional physical input in order to specify $Q$.

\subsection{Spontaneous symmetry breaking}
\label{ssb_section}

An important feature of any QFT is the phenomenon of spontaneous symmetry breaking (SSB). In general, the criterion for a quantum symmetry to be spontaneously broken can be characterised by the theorem~\cite{Strocchi08}:
\begin{theorem}
A symmetry is spontaneously broken \hspace{2mm} $\Longleftrightarrow$ \hspace{2mm} $\exists \, \varphi \in \mathcal{F}$ \ such that \ $\langle \delta \varphi \rangle \neq 0$ 
\label{ssb_cond} 
\end{theorem}
\noindent
where $\mathcal{F}$ is the (local) space of fields in the theory, $\delta$ is the symmetry variation defined in Eq.~(\ref{var_def}), and $\langle \cdot \rangle$ is the vacuum expectation value. In light of Corollary~\ref{nm_corr1}, it is important to establish whether there is a subsequent ambiguity in determining whether a non-manifest symmetry is spontaneously broken or not. As a consequence of Theorems~\ref{theorem_op_gen} and~\ref{ssb_cond}, one has the following corollary:
\begin{corr}
The criterion for a non-manifest symmetry to be spontaneously broken is independent of the choice of charge
\label{SSB_inv}
\end{corr}
\noindent
This means that although the charges $Q$ of non-manifest symmetries are in general not unique, since it is always possible to perform an improvement transformation $Q \rightarrow \widetilde{Q}$, one can equally use any of these charges to establish whether SSB occurs (via Theorem~\ref{ssb_cond}) without ambiguity. \\

\noindent    
Often SSB is characterised by the action of the charge $Q$ on the vacuum state, and in particular that SSB occurs if and only if $Q|0\rangle \neq 0$~\cite{Weinberg00,West86}. However, the problem with this condition is that unlike Theorem~\ref{ssb_cond}, the action of $Q$ on the vacuum state is \textit{not} invariant under improvement transformations $Q \rightarrow \widetilde{Q}$, and is therefore ambiguous for non-manifest symmetries. In fact, due to Corollary~\ref{Q_vacuum}, knowledge of $Q|0\rangle$ is required in order to uniquely define $Q$ in the first place. Therefore, if it were true that SSB could be solely characterised by $Q|0\rangle$, then this would mean that the physical input required to define $Q$ would automatically \textit{also} determine whether the symmetry is spontaneously broken or not. But this cannot be the case, because this would imply that every non-manifest symmetry could only either always be broken or always unbroken, but not both.

\section{Examples of quantum non-manifest symmetries}
\label{section3}

As outlined in section~\ref{section2}, quantum non-manifest symmetries have many interesting and subtle features. In this section we will discuss these features in the context of some prominent examples of these symmetries.

\subsection{Translational invariance}
\label{em_nm}

The invariance of a QFT under spacetime translations is an important example of a quantum non-manifest symmetry. In this case, the conserved current is the energy-momentum tensor $T^{\mu\nu}$ and the corresponding charge is the energy-momentum operator $P^{\mu}$. By applying the results of section~\ref{section2}, and in particular Theorem~\ref{theorem1} and Corollary~\ref{nm_corr1}, it follows that $T^{\mu\nu}$, and hence $P^{\mu}$, are ambiguously defined. Nevertheless, it is frequently cited in the literature~\cite{Shifman12,Weinberg95,Greiner_Reinhardt96} that one can always add an improvement term to $T^{\mu\nu}$ without modifying the corresponding charge. A prominent example of this is the symmetric Belinfante energy-momentum tensor $T^{\mu\nu}_{B}$ and the canonical current $T^{\mu\nu}_{c}$, which are related (by an improvement term) as follows~\cite{Belinfante40}:
\begin{align}
T^{\mu\nu}_{B} = T^{\mu\nu}_{c} + \partial_{\rho}G^{[\mu\rho]\nu} \,.
\end{align}
Now although Theorem~\ref{theorem_op_gen} implies that the corresponding charges $P^{\mu}_{B}$ and $P^{\mu}_{c}$ are both generators of translations, it is not necessarily the case that these operators act in the same manner on states\footnote{This potential difference between $P^{\mu}_{B}$ and $P^{\mu}_{c}$ has been largely ignored in the literature, but has been emphasised before in~\cite{Leader_Lorce14}.}. This in fact highlights a deeper problem -- \textit{how does one determine which energy-momentum operator is correct?} It is clear that in order to answer this question one requires additional physical information\footnote{It is sometimes concluded that the Belinfante generator is more physically motivated because $T^{\mu\nu}$ is symmetric when defined as the variational derivative of the action with respect to the metric in General Relativity. However, as pointed out in~\cite{Leader_Lorce14}, $T^{\mu\nu}$ need not be symmetric if one loosens the requirement that $g^{\mu\nu}$ is symmetric and covariantly constant ($\nabla_{\sigma}g^{\mu\nu}=0$), as is the case in Einstein-Cartan theory.}. In QFTs, the energy-momentum operator $P^{\mu}$ plays a special role in characterising the vacuum state $|0\rangle$. In particular, axiomatic formulations of QFT assume $|0\rangle$ to be the unique translationally invariant state~\cite{Streater_Wightman64,Haag96,Nakanishi_Ojima90}, which means it satisfies the condition: $P^{\mu}|0\rangle=0$. Since Corollary~\ref{Q_vacuum} implies that $P^{\mu}$ is uniquely defined by its action on $|0\rangle$, this condition provides a solution to the problem of which energy-momentum operator is physically relevant. In other words, if one can demonstrate that a certain expression for $P^{\mu}$ (e.g. $P^{\mu}_{B}$ or $P^{\mu}_{c}$) annihilates the vacuum, then this is sufficient to prove that this is the only charge that satisfies this property, and must therefore be \textit{the} physical energy-momentum operator.

\subsection{Lorentz invariance}

Invariance under Lorentz transformations is another example of a non-manifest symmetry. The conserved current is $M^{\mu\nu\lambda}$, and the corresponding charge is $M^{\mu\nu}$. Just like with the energy-momentum tensor, one has both Belinfante $M^{\mu\nu\lambda}_{B} = x^{\nu}T^{\mu\lambda}_{B} - x^{\lambda}T^{\mu\nu}_{B}$ and canonical $M^{\mu\nu\lambda}_{c}$ currents which are both conserved, and differ by an improvement term. Similarly, it follows from the conclusions in section~\ref{section2} that the charges $M^{\mu\nu}_{B}$ and $M^{\mu\nu}_{c}$ are both generators of Lorentz transformations, but are not necessarily the same operator. This again raises the same problem of how to establish which charge is physically relevant. As discussed in section~\ref{em_nm}, axiomatic formulations of QFT assume that the vacuum state $|0\rangle$ is the unique translationally invariant state. It transpires that this assumption implies that $|0\rangle$ is \textit{also} invariant under Lorentz transformations~\cite{Strocchi13}, and hence $M^{\mu\nu}|0\rangle =0$. Due to Corollary~\ref{Q_vacuum}, this physical condition therefore provides a way in which $M^{\mu\nu}$ can be uniquely determined. It should be noted that although the conditions $P^{\mu}|0\rangle=0$ and $M^{\mu\nu}|0\rangle =0$ appear relatively simple, their verification is not necessarily straight-forward, especially in QFTs such as quantum chromodynamics (QCD) where the vacuum state has a non-trivial structure~\cite{Yndurain99}. Nevertheless, in principle these conditions could be verified using a non-perturbative approach such as lattice QFT. \\
  
\noindent
Much of the discussion in the literature regarding the current $M^{\mu\nu\lambda}$ centres around the construction of angular momentum operators $J^{i} = \frac{1}{2}  \epsilon^{ijk}\int d^{3}x \, M^{0jk}$. In particular, an open problem in QCD which has received both significant theoretical and experimental focus, is the question of whether the angular momentum operator $J_{\text{QCD}}$ has a meaningful decomposition into separate quark and gluon contributions~\cite{Leader_Lorce14,Jaffe_Manohar90,Wakamatsu14,Ji97}. There are many different proposed decompositions of $J_{\text{QCD}}$, but they all have in common the fact that they are constructed by adding improvement terms to the canonical $M^{\mu\nu\lambda}_{c}$ or Belinfante $M^{\mu\nu\lambda}_{B}$ Lorentz currents. Although it remains uncertain which (if any) of these decompositions is physically meaningful~\cite{Lowdon14,Leader_Lorce14}, this is directly related to issue of whether certain spatial boundary operators vanish or not. Ultimately, if the physical Lorentz charge $M^{\mu\nu}$ (where $M^{\mu\nu}|0\rangle =0$) could be determined, this would be a significant step in answering this question.  

\subsection{Supersymmetry}

Supersymmetry corresponds to an enlargement of the Poincar\'{e} group of spacetime symmetries, and is another prominent example of a non-manifest symmetry. Invariance under supersymmetric transformations implies a conserved current $S^{\mu}_{\alpha}$, which gives rise to a spinor-valued charge $Q_{\alpha}$. An important feature of supersymmetric QFTs is that unlike the operators $P^{\mu}$ and $M^{\mu\nu}$ in non-supersymmetric theories, there is no equivalent physical requirement as to how $Q_{\alpha}$ (or $Q_{\alpha}^{\dagger}$) should act on the vacuum state. If one were to similarly assume on physical grounds that $Q_{\alpha}|0\rangle=0$ (and $Q_{\alpha}^{\dagger}|0\rangle=0$), then this would imply:
\begin{align}
\langle \delta \varphi \rangle = 0, \hspace{5mm} \forall \varphi \in \mathcal{F}
\end{align} 
and hence any physical supersymmetric QFT would have to have unbroken supersymmetry. The problem with this criterion is that unlike Poincar\'{e} symmetry, SSB plays a particularly important role in the characterisation of physically realistic supersymmetric theories. The reason for this is that the supersymmetry algebra implies that every known particle must have a corresponding supersymmetric partner, with equal mass~\cite{Weinberg00,West86}. Since these additional particles have never been observed, it is concluded that supersymmetry must be spontaneously broken~\cite{Weinberg00,West86}, and so this rules out $Q_{\alpha}|0\rangle=0$ as a general physical criterion. So if such a criterion did exist then it would necessarily have to imply that $Q_{\alpha}|0\rangle$ is a non-vanishing state. But since there is no clear physical principle as to what state $Q_{\alpha}|0\rangle$ should be, it follows from Corollaries~\ref{nm_corr1} and~\ref{Q_vacuum} that: \\
\begin{corr}
The supersymmetric charge $Q_{\alpha}$ is non-unique
\label{susy_char_cor}
\end{corr} 
\ \\
\noindent
Due to Theorem~\ref{theorem_op_gen} and Corollary~\ref{SSB_inv}, the ambiguity in the definition of $Q_{\alpha}$ does not affect the generation of supersymmetric transformations, nor the determination of whether the supersymmetry is spontaneous broken or not. Nevertheless, because the structure of $Q_{\alpha}$ is not fixed, the action of $Q_{\alpha}$ on states is not uniquely defined. Ultimately, this means that the supersymmetric space of states cannot be constructed in a consistent manner, and this therefore provides a non-perturbative obstacle to the consistency of supersymmetric QFTs.

\section{Conclusions}
\label{section4}

Non-manifest symmetries play an important role in QFT, and yet the quantisation of these symmetries is still not fully understood. Although it well known that the ambiguity in the definition of the conserved currents associated with these symmetries provides the freedom to define different charges, often classical arguments are incorrectly employed to justify that these charge operators are physically equivalent. The central issue in this regard is to determine the conditions under which spatial boundary operators vanish. It turns out that for locally quantised theories, there in fact exists both a necessary and sufficient condition for when this class of operators vanishes. By applying this condition it follows that the charge operators of non-manifest symmetries are non-unique, but different expressions for the charge operator still generate the same symmetry transformations. In the context of SSB, these results ensure that in spite of the non-uniqueness of the charge, the criterion for SSB is not affected by this ambiguity. Nevertheless, the charge non-uniqueness is still potentially problematic because it remains unclear as to how the charge operator acts on states. A prominent example of this is the definition of the energy-momentum $P^{\mu}$ and Lorentz charges $M^{\mu\nu}$ associated with the non-manifest symmetries of spacetime translation and Lorentz invariance. In each case, both canonical and Belinfante charges can be defined, but it is unclear which of these operators (if either) is more physically relevant. A solution to this problem is to use the physical assumption that the vacuum is the unique translationally invariant state, because it follows that a knowledge of how these operators act on the vacuum is enough to uniquely define them. However, in the case of supersymmetry, there is no such physical requirement as to how the supersymmetric charge $Q_{\alpha}$ should act on the vacuum. So by contrast to $P^{\mu}$ and $M^{\mu\nu}$, $Q_{\alpha}$ is \textit{not} uniquely defined, and this therefore introduces a non-perturbative obstacle to the consistency of supersymmetric QFTs.

\section*{Acknowledgments}
I thank Thomas Gehrmann for useful discussions and input. This work was supported by the Swiss National Science Foundation (SNF) under contract CRSII2\_141847.

\end{document}